\documentclass{llncs} 
\usepackage{llncsdoc} 
\usepackage{enumerate}
\usepackage{graphicx} 
\usepackage{graphics, color}
\usepackage{amsmath} 
\usepackage{algorithm} 
\usepackage{algorithmic}
\usepackage{boxedminipage} 
\usepackage{amssymb}
\usepackage{verbatim}
\usepackage{hyperref} 
\usepackage{array}
\usepackage{multirow}
\usepackage{color}
\usepackage{xcolor}
\usepackage{changebar}

\newcommand{\BoundedSAT}{\ensuremath{\mathsf{BoundedSAT}}}

\newcommand{\UniformWitness}{\ensuremath{\mathsf{UniWit}}}

\newcommand{\killthis}[1]{}

\newcommand{\prob}{\ensuremath{\mathsf{Pr}}}
\newcommand{\expect}{\ensuremath{\mathsf{E}}}

\newcommand{\mc}[1]{\ensuremath{\mathcal{#1}}}
\newcommand{\Relsat}{\ensuremath{\mathsf{Relsat}}}
\newcommand{\NP}{\ensuremath{\mathsf{NP}}}
\newcommand{\SAT}{\ensuremath{\mathsf{SAT}}}
\newcommand{\sharpSATTool}{\ensuremath{\mathsf{sharpSAT}}}
\newcommand{\Cachet}{\ensuremath{\mathsf{Cachet}}}
\newcommand{\ApproxCount}{\ensuremath{\mathsf{ApproxCount}}}

\newcommand{\SampleCount}{\ensuremath{\mathsf{SampleCount}}}
\newcommand{\BPCount}{\ensuremath{\mathsf{BPCount}}}
\newcommand{\MBound}{\ensuremath{\mathsf{MBound}}}
\newcommand{\CDP}{\ensuremath{\mathsf{CDP}}}
\newcommand{\HybridMBound}{\ensuremath{\mathsf{Hybrid}}-\ensuremath{\mathsf{MBound}}}

\newcommand{\MiniCount}{\ensuremath{\mathsf{MiniCount}}}
\newcommand{\sharpSAT}{\ensuremath{\#\mathsf{SAT}}}
\newcommand{\sharpP}{\ensuremath{\#\mathsf{P}}}
\newcommand{\approxMC}{\ensuremath{\mathsf{ApproxMC}}}
\newcommand{\ApproxMC}{\ensuremath{\mathsf{ApproxMC}}}
\newcommand{\ApproxMCCore}{\ensuremath{\mathsf{ApproxMCCore}}}
\newcommand{\ComputeThreshold}{\ensuremath{\mathsf{ComputeThreshold}}}
\newcommand{\ComputeIterCount}{\ensuremath{\mathsf{ComputeIterCount}}}
\newcommand{\SearchTreeSampler}{\ensuremath{\mathsf{SearchTreeSampler}}}
\newcommand{\SE}{\ensuremath{\mathsf{SE}}}
\newcommand{\SampleSearch}{\ensuremath{\mathsf{SampleSearch}}}
\newcommand{\JVV}{\ensuremath{\mathsf{JVV}}}

\newcommand{\added}[1]{#1}

\begin{document}
\makeatletter
\renewcommand*{\@fnsymbol}[1]{\ensuremath{\ifcase#1\or *\or \dagger\or \ddagger\or
   \mathsection\or \mathparagraph\or \|\or **\or \dagger\dagger
   \or \ddagger\ddagger \else\@ctrerr\fi}}
\makeatother
\title{A Scalable Approximate Model Counter\thanks{Authors would like to thank Henry Kautz and  Ashish Sabhrawal for their valuable help in experiments, and  Tracy Volz for valuable comments on the earlier drafts.  Work supported in part by NSF grants CNS 1049862 and CCF-1139011,
by NSF Expeditions in Computing project ``ExCAPE: Expeditions in Computer
Augmented Program Engineering," by BSF grant 9800096, by a gift from
Intel, by a grant from Board of Research in Nuclear Sciences, India, and by the Shared University Grid at Rice funded by NSF under Grant EIA-0216467, and a partnership between Rice University, Sun Microsystems, and Sigma Solutions, Inc.
} \thanks{A longer version of this paper is available at \url{http://www.cs.rice.edu/CS/Verification/Projects/ApproxMC/}}}
\author{Supratik Chakraborty\inst{1} \and Kuldeep S. Meel\inst{2} \and Moshe Y. Vardi\inst{2}}
\institute{Indian Institute of Technology Bombay, India
\and
Department of Computer Science, Rice University}
\date{Jan.~29, 2013}

\maketitle
\begin{abstract}
\emph{Propositional model counting} ({\sharpSAT}), i.e., counting the
number of satisfying assignments of a propositional formula, is a
problem of significant theoretical and practical interest.  Due to the
inherent complexity of the problem, \emph{approximate model counting},
which counts the number of satisfying assignments to within given
tolerance and confidence level, was proposed as a practical
alternative to exact model counting. Yet, approximate model counting
has been studied essentially only theoretically.  The only reported
implementation of approximate model counting, due to Karp and Luby,
worked only for DNF formulas.  A few existing tools for CNF formulas
are \emph{bounding model counters}; they can handle realistic problem
sizes, but fall short of providing counts within given tolerance and
confidence, and, thus, are not approximate model counters.

We present here a novel algorithm, as well as a reference implementation,
that is the first scalable approximate model counter for CNF formulas.
The algorithm works by issuing a polynomial number of calls to a
SAT solver. Our tool, {\approxMC}, scales to formulas with tens of thousands of 
variables.  Careful experimental comparisons show that {\approxMC}
reports, with high confidence, bounds that are close to the exact count,
and also succeeds in reporting bounds with small tolerance
and high confidence in cases that are too large for computing exact
model counts.
\end{abstract}

\section{Introduction}\label{sec:intro}
Propositional model counting, also known as {\sharpSAT}, concerns 
counting the number of models (satisfying truth assignments) of a
given propositional formula.  This problem has been the subject of
extensive theoretical investigation since its introduction by
Valiant~\cite{valiant1979complexity} in 1979.  Several interesting 
applications of {\sharpSAT} have been studied in the context of 
probabilistic reasoning, planning, combinatorial design and other 
related fields~\cite{Roth1996,Bacchus04algorithmsand,domshlak2007}.  
In particular, probabilistic reasoning and inferencing have attracted
considerable interest in recent years~\cite{GSS08}, and stand to
benefit significantly from efficient propositional model counters.
Theoretical investigations of {\sharpSAT} have led to the discovery of 
deep connections in complexity theory 
\cite{Angluin1980,simon77,toda1989computational}: 
{\sharpSAT} is {\sharpP}-complete, where {\sharpP} is the set of
counting problems associated with decision problems in the complexity
class {\NP}.  Furthermore, $\mathsf{P}^{\sharpSAT}$, that is,
a polynomial-time machine with a {\sharpSAT} oracle, can solve all 
problems in the entire polynomial hierarchy. In fact, the 
polynomial-time machine only needs to make one {\sharpSAT} query to 
solve any problem in the polynomial hierarchy. This is strong evidence 
for the hardness of {\sharpSAT}.

In many applications of model counting, such as in probabilistic
reasoning, the exact model count may not be critically important, and
approximate counts are sufficient.  Even when exact model counts are
important, the inherent complexity of the problem may force one to
work with approximate counters in practice.  In~\cite{Stockmeyer83},
Stockmeyer showed that counting models within a specified tolerance
factor can be achieved in deterministic polynomial time using a
$\Sigma_2^p$-oracle.  Karp and Luby presented a fully polynomial
randomized approximation scheme for counting models of a DNF formula
\cite{KarpLuby1989}.  Building on Stockmeyer's result, Jerrum, Valiant 
and Vazirani~\cite{Jerr} showed that counting models of CNF formulas
within a specified tolerance factor can be solved in random polynomial 
time using an oracle for {\SAT}.  

On the implementation front, the earliest approaches to {\sharpSAT}
were based on DPLL-style {\SAT} solvers and computed exact counts.
These approaches consisted of incrementally counting the number of
solutions by adding appropriate multiplication factors after a partial
solution was found.  This idea was formalized by Birnbaum and
Lozinkii~\cite {Birnbaum1999} in their model counter {\CDP}.  Subsequent
model counters such as {\Relsat}~\cite{Bayardo97usingcsp},
{\Cachet}~\cite{Sang04combiningcomponent} and
{\sharpSATTool}~\cite{Thurley2006} improved upon this idea by using
several optimizations such as component caching, clause learning,
look-ahead and the like.  Techniques based on Boolean Decision
Diagrams and their variants~\cite{minato93,lobbing1996}, or d-DNNF
formulae~\cite{darwiche2004new}, have also been used to compute exact
model counts.  Although exact model counters have been successfully
used in small- to medium-sized problems, scaling to larger problem
instances has posed significant challenges in practice.  Consequently,
a large class of practical applications has remained beyond the reach
of exact model counters.

To counter the scalability challenge, more efficient techniques for
counting models approximately have been proposed.  These counters can
be broadly divided into three categories.  Counters in the first
category are called $(\varepsilon, \delta)$ counters, following Karp
and Luby's terminology~\cite{KarpLuby1989}.  Let $\varepsilon$ and
$\delta$ be real numbers such that $ 0 < \varepsilon \leq 1$ and $0 <
\delta \le 1$.  For every propositional formula $F$ with $\#F$ models,
an $(\varepsilon, \delta)$ counter computes a number that lies in the
interval $[(1+\varepsilon)^{-1}\#F, (1+\varepsilon)\#F]$ with
probability at least $1-\delta$.  We say that $\varepsilon$ is the
\emph{tolerance} of the count, and $1-\delta$ is its
\emph{confidence}.  The counter described in this paper and also that
due to Karp and Luby~\cite{KarpLuby1989} belong to this category.  The
approximate-counting algorithm of Jerrum et al.~\cite{Jerr} also
belongs to this category; however, their algorithm does not lend
itself to an implementation that scales in practice.  Counters in the
second category are called \emph{lower (or upper) bounding counters},
and are parameterized by a confidence probability $1-\delta$.  For every
propositional formula $F$ with $\#F$ models, an upper (resp., lower)
bounding counter computes a number that is at least as large (resp.,
as small) as $\#F$ with probability at least $1 - \delta$.  Note that
bounding counters \emph{do not} provide any tolerance guarantees.  The
large majority of approximate counters used in practice are bounding
counters.  Notable examples include
{\SampleCount}~\cite{gomes2007sampling},
{\BPCount}~\cite{KrocSabSel2008}, {\MBound} (and
{\HybridMBound})~\cite{gomes2006model}, and
{\MiniCount}~\cite{KrocSabSel2008}.  The final category of counters is
called \emph{guarantee-less counters}.  These counters provide no
guarantees at all but they can \added{be} very efficient and provide good
approximations in practice.  \added{Examples of guarantee-less counters include
{\ApproxCount}~\cite{wei2005new}, {\SearchTreeSampler}~\cite{ErGoSel2012}, {\SE}~\cite{Rubin2012} and {\SampleSearch}~\cite{GogDech2011}.}

Bounding both the tolerance and confidence of approximate model counts
is extremely valuable in applications like probabilistic inference.
Thus, designing $(\varepsilon, \delta)$ counters that scale to
practical problem sizes is an important problem.  Earlier work on
$(\varepsilon, \delta)$ counters has been restricted largely to
theoretical treatments of the problem.  The only counter in this
category that we are aware of as having been implemented is due to
Karp and Luby~\cite{LubyThesis83}.  Karp and Luby's original
implementation was designed to estimate reliabilities of networks with
failure-prone links.  However, the underlying Monte Carlo engine can
be used to approximately count models of DNF, \emph{but not CNF},
formulas.

The counting problems for both CNF and DNF formulae are
{\sharpP}-complete. While the DNF representation suits some
applications, most modern applications of model counting (e.g.
probabilistic inference) use the CNF representation.  Although
\emph{exact} counting for DNF and CNF formulae are polynomially
inter-reducible, there is no known polynomial reduction for the
corresponding \emph{approximate} counting problems.  In fact, Karp and
Luby remark in~\cite{KarpLuby1989}  that it is highly unlikely that their
randomized approximate algorithm for DNF formulae can be adapted to
work for CNF formulae.  Thus, there has been no prior implementation
of $(\varepsilon, \delta)$ counters for CNF formulae \emph{that scales in
practice}.  In this paper, we present the first such counter. As in
\cite{Jerr}, our algorithm runs in random polynomial time using an
oracle for {\SAT}. Our extensive experiments show that our algorithm
scales, with low error, to formulae arising from several application domains involving tens of thousands of
variables.

The organization of the paper is as follows. We present preliminary
material in Section \ref{sec:prelims}, and related work in Section
\ref{sec:relatedwork}. In Section \ref{sec:algo}, we present our
algorithm, followed by its analysis in Section \ref{sec:analysis}.
Section \ref{sec:experiment} discusses our experimental methodology,
followed by experimental results in Section \ref{sec:results}. Finally, we conclude in Section \ref{sec:discussion}.

\section{Notation and Preliminaries}\label{sec:prelims}

Let $\Sigma$ be an alphabet and $R \subseteq \Sigma^* \times \Sigma^*$
be a binary relation.  We say that $R$ is an \mc{NP}-relation if $R$
is polynomial-time decidable, and if there exists a polynomial
$p(\cdot)$ such that for every $(x, y) \in R$, we have $|y| \le
p(|x|)$.  Let $L_R$ be the language $\{x \in \Sigma^* \mid \exists y
\in \Sigma^*,\, (x, y) \in R\}$.  The language $L_R$ is said to be in
\mc{NP} if $R$ is an \mc{NP}-relation.  The set of all satisfiable
propositional logic formulae in CNF is a language in \mc{NP}.  
Given $x \in L_R$, a \emph{witness} or \emph{model} of $x$
is a string $y \in \Sigma^*$ such that $(x, y) \in R$.  The set of all
models of $x$ is denoted $R_x$.  For notational convenience,
fix $\Sigma$ to be $\{0, 1\}$ without loss of generality.  If $R$ is
an \mc{NP}-relation, we may further assume that for every $x \in L_R$,
every witness $y \in R_x$ is in $\{0, 1\}^n$, where $n = p(|x|)$ for
some polynomial $p(\cdot)$.

Let $R \subseteq \{0,1\}^* \times \{0,1\}^*$ be an \mc{NP} relation.
The \emph{counting problem} corresponding to $R$ asks ``Given $x \in
\{0,1\}^*$, what is $|R_x|$?''.  If $R$ relates CNF propositional
formulae to their satisfying assignments, the corresponding counting
problem is called {\sharpSAT}.  The primary focus of this paper is on
$(\varepsilon, \delta)$ counters for {\sharpSAT}.  The randomized
$(\varepsilon, \delta)$ counters of Karp and Luby~\cite{KarpLuby1989}
for DNF formulas are \emph{fully polynomial}, which means that
they run in time \added{polynomial} in the size of the input 
formula $F$, $1/\varepsilon$ and $\log(1/\delta)$.  The randomized 
$(\varepsilon, \delta)$ counters for CNF formulas in \cite{Jerr}
and in this paper are however fully polynomial \emph{with respect to a SAT oracle}.

A special class of hash functions, called \emph{$r$-wise independent} 
hash functions, play a crucial role in our work.
Let $n, m$ and $r$ be positive integers, and let $H(n,m,r)$ denote a
family of $r$-wise independent hash functions mapping $\{0, 1\}^n$ to $\{0, 1\}^m$.  We use 
$\prob\left[X: {\cal P} \right]$ to denote the probability of outcome 
$X$ when sampling from a probability space ${\cal P}$, and 
$h \xleftarrow{R} H(n,m,r)$ to denote the probability space obtained by
choosing a hash function $h$ uniformly at random from $H(n,m,r)$. The 
property of $r$-wise independence guarantees 
that for all $\alpha_1, \ldots \alpha_r \in \{0,1\}^m $ and
for all distinct $y_1, \ldots y_r \in \{0,1\}^n$, 
$\prob\left[\bigwedge_{i=1}^r h(y_i) = \alpha_i\right.$ $\left.: 
h \xleftarrow{R} H(n, m, r)\right] = 2^{-mr}$.  For every $\alpha \in
\{0, 1\}^m$ and $h \in H(n, m, r)$, let $h^{-1}(\alpha)$ denote the
set $\{y \in \{0, 1\}^n \mid h(y) = \alpha\}$.  Given $R_x \subseteq
\{0, 1\}^n$ and $h \in H(n, m, r)$, we use $R_{x, h, \alpha}$ to
denote the set $R_x \cap h^{-1}(\alpha)$.  If we keep $h$ fixed and
let $\alpha$ range over $\{0, 1\}^m$, the sets $R_{x, h, \alpha}$ form
a partition of $R_x$.  Following the notation in ~\cite{Bellare98uniformgeneration}, we call each element of
such a \added{partition} a \emph{cell} of $R_x$ \emph{induced} by $h$.  It was 
shown in~\cite{Bellare98uniformgeneration} that if $h$ is chosen
uniformly at random from $H(n, m,r)$ for $r \ge 1$, then the expected 
size of $R_{x,h,\alpha}$, 
denoted $\expect\left[|R_{x,h,\alpha}|\right]$,
is $|R_x|/2^m$, for each $\alpha \in \{0, 1\}^m$.

 The specific family of hash functions used in our work, denoted
 $H_{xor}(n, m, 3)$, is based on randomly choosing bits from $y \in
 \{0, 1\}^n$ and xor-ing them. This family of hash functions has been
 used in earlier work ~\cite{gomes2006model}, and has been shown to be
 3-independent in ~\cite{Gomes-Sampling}.  Let $h(y)[i]$ denote the
 $i^{th}$ component of the bit-vector obtained by applying hash
 function $h$ to $y$.  The family $H_{xor}(n, m, 3)$ is defined as
 $\{h(y) \mid (h(y))[i] = a_{i,0} \oplus (\bigoplus_{k=1}^n
 a_{i,k}\cdot y[k]), a_{i,j} \in \{0, 1\}, 1 \leq i \leq m, 0 \leq j
 \leq n\}$, where $\oplus$ denotes the xor operation.  By randomly
 choosing the $a_{i,j}$'s, we can randomly choose a hash function from
 this family.

\section{Related Work}\label{sec:relatedwork}

Sipser pioneered a hashing based approach in ~\cite{Sipser83}, which
has subsequently been used in theoretical
~\cite{trevisan2002lecture,Bellare98uniformgeneration} and practical
~\cite{Gomes-Sampling,gomes2006model,SKV13} treatments of approximate counting and
(near-)uniform sampling.
Earlier implementations of counters that use the hashing-based
approach are {\MBound} and {\HybridMBound}~\cite{gomes2006model}.
Both these counters use the same family of hashing functions, i.e.,
$H_{xor}(n, m, 3)$, that we use.  Nevertheless, there are significant
differences between our algorithm and those of {\MBound} and
{\HybridMBound}.  Specifically, we are able to exploit properties of
the $H_{xor}(n, m, 3)$ family of hash functions to obtain a fully
polynomial $(\varepsilon, \delta)$ counter with respect to a SAT
oracle.  In contrast, both {\MBound} and {\HybridMBound} are bounding
counters, and cannot provide bounds on tolerance.  In addition, our
algorithm requires no additional parameters beyond the tolerance
$\varepsilon$ and confidence $1-\delta$.  In contrast, the performance
and quality of results of both {\MBound} and {\HybridMBound}, depend
crucially on some hard-to-estimate parameters.  It has been our
experience that the right choice of these parameters is often domain
dependent \added{and difficult}.

Jerrum, Valiant and Vazirani~\cite{Jerr} showed that if $R$ is a 
self-reducible \mc{NP} relation (such as SAT),
the problem of generating models \emph{almost uniformly} is 
polynomially inter-reducible with approximately counting models.  The 
notion of almost uniform generation requires that if $x$ 
is a problem instance, then for every $y \in R_x$, we have $(1
+ \varepsilon)^{-1}\varphi(x)$ $\le$ $\prob[y \mbox{ is generated}]$
$\le$ $(1 + \varepsilon)\varphi(x)$, where 
$\varepsilon > 0$ is the specified tolerance and $\varphi(x)$ is an
appropriate function. Given an almost 
uniform generator $\mc{G}$ for $R$, an input $x$, a tolerance bound
$\varepsilon$ and an error probability bound $\delta$, it is shown
in \cite{Jerr} that one can obtain an $(\varepsilon, \delta)$
counter for $R$ by invoking $\mc{G}$ polynomially (in $|x|$,
$1/\varepsilon$ and $\log_2(1/\delta)$) many times, and by using the
generated samples to estimate $|R_x|$.  For convenience of exposition,
we refer to this approximate-counting algorithm as the {\JVV}
algorithm (after the last names of the authors).

An important feature of the {\JVV} algorithm is that it uses the
almost uniform generator $\mc{G}$ as a black box.  Specifically, the
details of how $\mc{G}$ works is of no consequence.  Prima facie, this
gives us freedom in the choice of $\mc{G}$ when implementing the
{\JVV} algorithm.  Unfortunately, while there are theoretical 
constructions of uniform 
generators in~\cite{Bellare98uniformgeneration},
we are not aware of any implementation of an almost uniform generator 
that scales to CNF formulas involving thousands of variables.
The lack of a scalable and almost uniform generator presents
a significant hurdle in implementing the {\JVV} algorithm for
practical applications.
It is worth asking if we can make the {\JVV} algorithm work without
requiring $\mc{G}$ to be an almost uniform generator.  A closer look
at the proof of correctness of the {\JVV} algorithm~\cite{Jerr} shows
it relies crucially on the ability of $\mc{G}$ to ensure that the
probabilities of generation of any two distinct models of $x$ differ
by a factor in $O(\varepsilon^2)$.  As discussed in~\cite{SKV13},
existing algorithms for randomly generating models either provide this
guarantee but scale very poorly in practice (e.g., the algorithms
in~\cite{Bellare98uniformgeneration,Yuan2004}), or scale well in
practice without providing the above guarantee (e.g., the algorithms
in~\cite{SKV13,Gomes-Sampling,KitKue2007}).  Therefore, using an
existing generator as a black box in the {\JVV} algorithm would not
give us an $(\varepsilon, \delta)$ model counter that scales in
practice. %
The primary contribution of this paper is to show that a scalable
$(\varepsilon, \delta)$ counter can indeed be designed by using the
same insights that went into the design of a \emph{near uniform}
generator, {\UniformWitness}~\cite{SKV13}, but without using the
generator as a black box in \added{the} approximate counting
algorithm.  Note that near uniformity, as defined in~\cite{SKV13}, is
an even more relaxed notion of uniformity than almost uniformity.  We
leave the question of whether a near uniform generator can be used as
a black box to design an $(\varepsilon, \delta)$ counter as part of
future work.

 The central idea of {\UniformWitness}, which is also shared by our
 approximate model counter, is the use of $r$-wise independent hashing
 functions to randomly partition the space of all models of a given
 problem instance into ``small'' cells. This idea was first proposed
 in~\cite{Bellare98uniformgeneration}, but there are two novel
 insights that allow {\UniformWitness}~\cite{SKV13} to scale better
 than other hashing-based sampling
 algorithms~\cite{Bellare98uniformgeneration,Gomes-Sampling}, while
 still providing guarantess on the quality of sampling.  These
 insights are: (i) the use of computationally efficient linear hashing
 functions with low degrees of independence, and (ii) a drastic
 reduction in the size of ``small'' cells, from $n^2$
 in~\cite{Bellare98uniformgeneration} to $n^{1/k}$ (for $2 \le k \le
 3$) in~\cite{SKV13}, and even further to a constant in the current
 paper. %
We continue to use these key insights in the design of our approximate model counter,
although {\UniformWitness} is not used explicitly in the model counter.

\section{Algorithm} \label{sec:algo}
We now describe our approximate model counting algorithm, called
{\ApproxMC}.  %
As mentioned above, we use $3$-wise independent linear hashing
functions from the $H_{xor}(n, m, 3)$ family, for an appropriate $m$,
to randomly partition the set of models of an input formula into
``small'' cells.  In order to test whether the generated cells are
indeed small, we choose a random cell and check if it is non-empty and
has no more than $pivot$ elements, where $pivot$ is a threshold that
depends only on the tolerance bound $\varepsilon$.  If the chosen cell
is not small, we randomly partition the set of models into twice as
many cells as before by choosing a random hashing function from the
family $H_{xor}(n, m+1, 3)$.
The above procedure is repeated until either a randomly
chosen cell is found to be non-empty and small, or the number of
cells exceeds $\frac{2^{n+1}}{\mathit{pivot}}$.  If all cells
that were randomly chosen during the above process were either empty
or not small, we report a counting failure and return $\bot$.
Otherwise, the size of the cell last chosen is scaled by the number
of cells to obtain an $\varepsilon$-approximate estimate of the model count.  

The procedure outlined above forms the core engine of {\ApproxMC}. For
convenience of exposition, we implement this core engine as a function
{\ApproxMCCore}.  The overall {\ApproxMC} algorithm simply invokes
{\ApproxMCCore} sufficiently many times, and returns the median of the
non-$\bot$ values returned by {\ApproxMCCore}.  The pseudocode for
algorithm {\ApproxMC} is shown below.

\begin{tabular}{lccl}
\begin{minipage}{0.5\textwidth}
\begin{tabbing}	
xx \= xx \= xx \= xx \= xx \= \kill
\\
\noindent {\bfseries \textsf{Algorithm} ${\approxMC} (F, \varepsilon,\delta)$}\\
1:\> $\mathit{counter} \leftarrow 0; C \leftarrow \mathsf{emptyList}$;\\
2:\> $\mathit{pivot} \leftarrow 2 \times {\ComputeThreshold}(\varepsilon)$; \\
3:\> $t \leftarrow  {\ComputeIterCount}(\delta)$;\\
4:\> {\bfseries repeat}:\\
5:\> \> $c \leftarrow {\ApproxMCCore} (F, \mathit{pivot})$;\\
6:\> \> $\mathit{counter} \leftarrow \mathit{counter}+1$;\\
7:\> \> {\bfseries if} $(c \neq \bot)$ \\
8:\> \> \>$\mathsf{AddToList}(C, c)$;\\
9:\> {\bfseries until} ($\mathit{counter} < t$);\\
10:\> $\mathit{finalCount} \leftarrow \mathsf{FindMedian}(C)$;\\
11:\>{\bfseries return} $\mathit{finalCount}$; \\ 
\end{tabbing} 
\end{minipage} & & &
\begin{minipage}{0.5\textwidth}
\begin{tabbing}	
xx \= xx \= xx \= xx \= xx \= \kill
\\
\noindent {\bfseries \textsf{Algorithm} ${\ComputeThreshold} (\varepsilon)$}\\
1:\> {\bfseries return} $\left\lceil 3e^{1/2}\left(1 + \frac{1}{\varepsilon}\right)^2 \right\rceil$;\\
\\
\\
\noindent {\bfseries \textsf{Algorithm} ${\ComputeIterCount} (\delta)$}\\
1:\> {\bfseries return} $\left\lceil 35\log_2 (3/\delta) \right\rceil$;\\
\\
\\
\\
\\
\\
\end{tabbing}
\end{minipage}\\
\end{tabular}

\noindent
Algorithm {\ApproxMC} takes as inputs a CNF formula $F$, a tolerance
$\varepsilon ~(0 < \varepsilon \le 1)$ and $\delta ~(0 < \delta \le
1)$ such that the desired confidence is $1-\delta$.  It computes two
key parameters: (i) a threshold $\mathit{pivot}$ that depends only on
$\varepsilon$ and is used in {\ApproxMCCore} to determine the size of
a ``small'' cell, and (ii) a parameter $t ~(\ge 1)$ that depends only
on $\delta$ and is used to determine the number of times
{\ApproxMCCore} is invoked. The particular choice of functions to
compute the parameters $\mathit{pivot}$ and $t$ aids us in proving
theoretical guarantees for {\ApproxMC} in Section
\ref{sec:analysis}. Note that $\mathit{pivot}$ is in
$\mathcal{O}(1/\varepsilon^2)$ and $t$ is in $\mathcal{O}(\log_2 (1/\delta))$.  All non-$\bot$
estimates of the model count returned by {\ApproxMCCore} are stored in
the list $C$.  The function $\mathsf{AddToList}(C, c)$ updates the
list $C$ by adding the element $c$.  The final estimate of the model
count returned by {\ApproxMC} is the median of the estimates stored in
$C$, computed using $\mathsf{FindMedian}(C)$.  We assume that if the list $C$ is
empty, $\mathsf{FindMedian}(C)$ returns $\bot$.

The pseudocode for algorithm {\ApproxMCCore} is shown below.
\begin{tabbing} xx \= xx \= xx \= xx \= xx \= \kill
\noindent {\bfseries \textsf{Algorithm} ${\ApproxMCCore} (F, pivot)$}\\ 
/* Assume $z_1, \ldots z_n$ are the variables of $F$    */ \\
1:\> $S \leftarrow \BoundedSAT(F, \mathit{pivot} +1)$;\\
2:\> {\bfseries if} ($|S| \le \mathit{pivot}$)\\
3:\> \> return $|S|$; \\
4:\> {\bfseries else} \\
5:\> \> $l \leftarrow \lfloor \log_2 (\mathit{pivot}) \rfloor - 1$; $i \leftarrow l - 1$;\\
6:\> \> {\bfseries repeat}\\
7:\> \> \> $i \leftarrow i+1$;\\
8:\> \> \> Choose $h$ at random from $H_{xor}(n, i-l, 3)$; \\
9:\> \> \> Choose $\alpha$ at random from $\{0, 1\}^{i-l}$;\\
10:\> \> \> $S \leftarrow \BoundedSAT(F \wedge (h(z_1, \ldots z_n) = \alpha), \mathit{pivot}+1)$;\\
11:\> \> {\bfseries until} ($1 \le |S| \le \mathit{pivot}$) or ($i = n$);\\
12:\> \> {\bfseries if} ($|S| > \mathit{pivot}$ {\bfseries or} $|S| = 0$) {\bfseries return} $\bot$  ;\\
13:\> \> {\bfseries else return }  $|S| \cdot 2^{i-l}$;\\
\end{tabbing}
Algorithm {\ApproxMCCore} takes as inputs a CNF formula $F$ and a
threshold $pivot$, and returns an $\varepsilon$-approximate estimate
of the model count of $F$.  We assume that {\ApproxMCCore} has access
to a function {\BoundedSAT} that takes as inputs a proposition formula
$F'$ that is the conjunction of a CNF formula and xor constraints, as
well as a threshold \added{$v \ge 0$.  {\BoundedSAT}$(F', v)$} returns
a set $S$ of models of $F'$ such that $|S| = \min(\added{v}, \#F')$.
If the model count of $F$ is no larger than $\mathit{pivot}$, then
{\ApproxMCCore} returns the exact model count of $F$ in line $3$ of
the pseudocode.  Otherwise, it partitions the space of all models of
$F$ using random hashing functions from $H_{xor}(n, i-l, 3)$ and
checks if a randomly chosen cell is non-empty and has at most
$\mathit{pivot}$ elements.  Lines $8$--$10$ of the repeat-until loop
in the pseudocode implement this functionality.  The loop terminates
if either a randomly chosen cell is found to be small and non-empty,
or if the number of cells generated exceeds
$\frac{2^{n+1}}{\mathit{pivot}}$ (if $i=n$ in line $11$, the number of
cells generated is $2^{n-l} \ge \frac{2^{n+1}}{\mathit{pivot}}$).  In
all cases, unless the cell that was chosen last is empty or not small,
we scale its size by the number of cells generated by the
corresponding hashing function to compute an estimate of the model
count.  If, however, all randomly chosen cells turn out to be empty or
not small, we report a counting error by returning $\bot$.

\noindent 
{\bfseries Implementation issues:} There are two steps in
algorithm {\ApproxMCCore} (lines 8 and 9 of the pseudocode) where
random choices are made.  Recall from
Section~\ref{sec:prelims} that choosing a random hash
function from $H_{xor}(n, m, 3)$ requires choosing random
bit-vectors.  It is straightforward to implement these
choices and also the choice of a random $\alpha \in \{0, 1\}^{i-l}$ 
in line 9 of the pseudocode, if we have access to
a source of independent and uniformly distributed random bits.  
Our implementation uses pseudo-random sequences of bits generated from
nuclear decay processes and made available at HotBits~\cite{HotBits}.
We download and store a sufficiently long sequence of random bits in a
file, and access an appropriate number of bits sequentially whenever
needed. We defer experimenting with sequences of bits obtained from
other pseudo-random generators to a future study.

In lines 1 and 10 of the pseudocode for algorithm {\ApproxMCCore}, we
invoke the function {\BoundedSAT}.  Note that if $h$ is chosen
randomly from $H_{xor}(n, m, 3)$, the formula for which we seek models
is the conjunction of the original (CNF) formula and xor constraints
encoding the inclusion of each witness in $h^{-1}(\alpha)$.  We
therefore use a SAT solver optimized for conjunctions of xor
constraints and CNF clauses as the back-end engine.  Specifically, we
use CryptoMiniSAT (version 2.9.2)~\cite{CryptoMiniSAT}, which also
allows passing a parameter indicating the maximum number of witnesses
to be generated.

Recall that {\ApproxMCCore} is invoked $t$ times with the same
arguments in algorithm {\ApproxMC}.  Repeating the loop of lines 6--11
in the pseudocode of {\ApproxMCCore} in each invocation can be time
consuming if the values of $i-l$ for which the loop terminates are
large.  In~\cite{SKV13}, a heuristic called \emph{leap-frogging} was
proposed to overcome this bottleneck in practice.  With leap-frogging,
we register the smallest value of $i-l$ for which the loop terminates
during the first few invocations of {\ApproxMCCore}.  In all
subsequent invocations of {\ApproxMCCore} with the same arguments, we
start iterating the loop of lines 6--11 by initializing $i-l$ to the
smallest value registered from earlier invocations.  Our experiments
indicate that leap-frogging is extremely efficient in practice and
leads to significant savings in time after the first few invocations
of {\ApproxMCCore}. \added{A theoretical analysis of leapfrogging is
  deferred to future work.}

\section{Analysis of {\ApproxMC}}\label{sec:analysis}

The following result, a minor variation of Theorem~5 
in~\cite{Schmidt}, about Chernoff-Hoeffding bounds plays an important
role in our analysis.
\begin{theorem}\label{theorem:chernoff-hoeffding}
Let $\Gamma$ be the sum of $r$-wise independent random variables, each
of which is confined to the interval $[0, 1]$, and suppose
$\expect[\Gamma] = \mu$.  For $0 < \beta \le 1$, if $r \le \left\lfloor
\beta^{2}\mu e^{-1/2} \right\rfloor \leq 4$ , then $\prob\left[\,|\Gamma - \mu| \ge
  \beta\mu\,\right] \le e^{-r/2}$.
\end{theorem}

Let $F$ be a CNF propositional formula with $n$ variables. The next
two lemmas show that algorithm {\ApproxMCCore}, when invoked from
{\ApproxMC} with arguments $F$, $\varepsilon$ and $\delta$, behaves
like an $(\varepsilon, d)$ model counter for $F$, for a fixed confidence
$1-d$ (possibly different from $1-\delta$).  Throughout this
section, we use the notations $R_F$ and $R_{F,h,\alpha}$ introduced in
Section~\ref{sec:prelims}.

\begin{lemma}\label{lm:probProof}
 Let algorithm {\ApproxMCCore}, when invoked from {\ApproxMC}, return
 $c$ with $i$ being the final value of the loop counter in {\ApproxMCCore}.  
Then, $\prob\left[(1 + \varepsilon)^{-1}\cdot |R_F| \le c \le (1 +
   \varepsilon)\cdot |R_F| \Bigm|  c \neq \bot \mbox{ and } i \leq \log_2 |R_F|\right]$ $\ge 1 - e^{-3/2}$.
\end{lemma}
\begin{proof}
Referring to the pseudocode of {\ApproxMCCore}, the lemma is trivially
satisfied if $|R_F| \le \mathit{pivot}$.  Therefore, the only
non-trivial case to consider is when $|R_F| > \mathit{pivot}$ and
{\ApproxMCCore} returns from line $13$ of the pseudocode.  In this
case, the count returned is $2^{i-l}.|R_{F,h,\alpha}|$, where $l =
\lfloor \log_2 (\mathit{pivot}) \rfloor - 1$ and $\alpha, i$ and $h$
denote (with abuse of notation) the values of the corresponding
variables and hash functions in the final iteration of the
repeat-until loop in lines $6$--$11$ of the pseudocode.  

For simplicity of exposition, we assume henceforth that $\log_2
(\mathit{pivot})$ is an integer.  A more careful analysis removes this
restriction with only a constant factor scaling of the probabilities.
From the pseudocode of {\ApproxMCCore}, we know that
$\mathit{pivot} =2 \left\lceil 3e^{1/2}\left(1 +
\frac{1}{\varepsilon}\right)^2 \right\rceil$.  

Furthermore, the value of $i$ is always in $\{l, \ldots n\}$.  Since
$\mathit{pivot} < |R_F| \le 2^n$ and $l = \lfloor \log_2
\mathit{pivot} \rfloor - 1$, we have $l < \log_2 |R_F| \le n$.  The
lemma is now proved by showing that for every $i$ in $\{l, \ldots
\lfloor \log_2 |R_F| \rfloor\}$, $h \in H(n, i-l, 3)$ and $\alpha \in
\{0,1\}^{i-l}$, we have $\prob\left[(1 + \varepsilon)^{-1}\cdot |R_F|
  \le 2^{i-l}|R_{F,h,\alpha}|\right.$ $\left.\le (1 +
  \varepsilon)\cdot |R_F|\right]$ $\ge (1 - e^{-3/2})$.

For every $y \in \{0, 1\}^n$ and for every $\alpha \in \{0,
1\}^{i-l}$, define an indicator variable $\gamma_{y, \alpha}$ as
follows: $\gamma_{y, \alpha} = 1$ if $h(y) = \alpha$, and
$\gamma_{y,\alpha} = 0$ otherwise.  Let us fix $\alpha$ and $y$ and
choose $h$ uniformly at random from $H(n, i-l, 3)$.  The random choice
of $h$ induces a probability distribution on $\gamma_{y, \alpha}$,
such that $\prob\left[\gamma_{y, \alpha} = 1\right] = \prob\left[h(y)
  = \alpha\right] = 2^{-(i-l)}$, and
$\expect\left[\gamma_{y,\alpha}\right] = \prob\left[\gamma_{y, \alpha}
  = 1\right] = 2^{-(i-l)}$.  In addition, the $3$-wise independence of
hash functions chosen from $H(n, i-l, 3)$ implies that for every
distinct $y_a, y_b, y_c \in R_F$, the random variables $\gamma_{y_a,
  \alpha}$, $\gamma_{y_b, \alpha}$ and $\gamma_{y_c, \alpha}$ are
$3$-wise independent.

Let $\Gamma_\alpha = \sum_{y \in R_F} \gamma_{y, \alpha}$ and
$\mu_\alpha = \expect\left[\Gamma_\alpha\right]$.  Clearly,
$\Gamma_\alpha = |R_{F, h, \alpha}|$ and $\mu_\alpha = \sum_{y \in
  R_F} \expect\left[\gamma_{y, \alpha}\right] = 2^{-(i-l)}|R_F|$.
Since $|R_F| > \mathit{pivot}$ and $i \leq \log_2 |R_F|$, using the expression for
$\mathit{pivot}$, we get $3 \le \left\lfloor e^{-1/2}(1 +
\frac{1}{\varepsilon})^{-2}\cdot\frac{|R_F|}{2^{i-l}}
\right\rfloor$. Therefore, using Theorem
\ref{theorem:chernoff-hoeffding},
$\prob\left[|R_F|.\left(1-\frac{\varepsilon}{1+\varepsilon}\right) \leq
  2^{i-l}|R_{F,h,\alpha}|\right.$ $\left.\leq
  (1+\frac{\varepsilon}{1+\varepsilon})|R_F|\right] \ge 1- e^{-3/2}$.
 Simplifying and noting that $\frac{\varepsilon}{1+\varepsilon} <
\varepsilon$ for all $\varepsilon > 0$, we obtain
$\prob\left[(1+\varepsilon)^{-1}\cdot |R_F| \leq \right.$ $\left. 
  2^{i-l}|R_{F,h,\alpha}|\leq (1+ \varepsilon)\cdot |R_F| \right] \ge 1- e^{-3/2}$.
\end{proof}

\begin{lemma}\label{lm:nonbotProb}
Given $|R_F| > \mathit{pivot}$, the probability that an invocation of
{\ApproxMCCore} from {\ApproxMC} returns non-$\bot$ with $i \le \log_2
|R_F|$, is at least $1-e^{-3/2}$.
\end{lemma}
\begin{proof}
Let us denote $\log_2 |R_F| - l$ $=$ $\log_2 |R_F| -
(\left\lfloor\log_2 (\mathit{pivot}) \right\rfloor - 1)$ by $m$.  
Since $|R_F| > \mathit{pivot}$ and $|R_F| \le 2^n$, we have $l < m+l \le n$.
Let $p_i~(l \le i \le n)$ denote the conditional probability that
{\ApproxMCCore}$(F, \mathit{pivot})$ terminates in iteration $i$ of
the repeat-until loop (lines $6$--$11$ of the pseudocode) with $1 \le
|R_{F,h,\alpha}| \le \mathit{pivot}$, given $|R_F| > \mathit{pivot}$.
Since the choice of $h$ and $\alpha$ in each iteration of the loop are
independent of those in previous iterations, the conditional
probability that {\ApproxMCCore}$(F, \mathit{pivot})$ returns
non-$\bot$ with $i \le \log_2 |R_F| = m+l$, given $|R_F| >
\mathit{pivot}$, is $p_l + (1-p_l)p_{l+1}$ $+ \cdots +
(1-p_l)(1-p_{l+1})\cdots(1-p_{m+l-1})p_{m+l}$. Let us denote this sum
by $P$.  Thus, $P = p_l + \sum_{i=l+1}^{m+l} \prod_{k=l}^{i-1}
(1-p_k)p_i$ $\,\ge\, \left(p_l + \sum_{i=l+1}^{m+l-1}
\prod_{k=l}^{i-1} (1-p_k)p_i\right)p_{m+l}$ $+$ $\prod_{s=l}^{m+l-1}
(1-p_s)p_{m+l}$ $= p_{m+l}$.  The lemma is now proved by using
Theorem~\ref{theorem:chernoff-hoeffding} to show that $p_{m+l} \ge
1-e^{-3/2}$.

It was shown in Lemma~\ref{theorem:chernoff-hoeffding} that
$\prob\left[(1+\varepsilon)^{-1}\cdot |R_F| \leq
  2^{i-l}|R_{F,h,\alpha}|\right.$ $\left.\leq (1+ \varepsilon)\cdot
  |R_F| \right] \ge 1- e^{-3/2}$ for every $i \in \{l, \ldots \lfloor \log_2
|R_F| \rfloor\}$, $h \in H(n, i-l, 3)$ and $\alpha \in \{0,1\}^{i-l}$.
Substituting $\log_2 |R_F| = m+l$ for $i$, re-arranging terms and
noting that the definition of $m$ implies $2^{-m}|R_F| =
\mathit{pivot}/2$, we get
$\prob\left[(1+\varepsilon)^{-1}(\mathit{pivot}/2)\right.$
  $\left. \leq |R_{F,h,\alpha}|\right.$ $\left.\leq (1+
  \varepsilon)(\mathit{pivot}/2) \right] \ge 1- e^{-3/2}$.  Since $0 <
\varepsilon \le 1$ and $\mathit{pivot} > 4$, it follows that
$\prob\left[1 \le |R_{F,h,\alpha}| \le \mathit{pivot}\right]$ $\ge$
$1-e^{-3/2}$.  Hence, $p_{m+l} \ge 1-e^{-3/2}$.
\end{proof}

\begin{theorem}\label{thm:almost-approx}
Let an invocation of {\ApproxMCCore} from {\ApproxMC} return $c$. Then
$\prob\left[c \neq \bot \mbox{ and } (1 + \varepsilon)^{-1}\cdot |R_F| \le c
  \le (1 + \varepsilon)\cdot |R_F|\right]$ $\ge (1 - e^{-3/2})^2 > 0.6$.
\end{theorem}
\noindent \emph{Proof sketch:} It is easy to see that the required
probability is at least as large as $\prob\left[c \neq \bot \mbox{ and
  } i \leq \log_2|R_F| \mbox{ and } (1 + \varepsilon)^{-1}\cdot |R_F|
  \le c \le (1 + \varepsilon)\cdot |R_F|\right]$. From
Lemmas~\ref{lm:probProof} and \ref{lm:nonbotProb}, the latter
probability is $\ge (1 - e^{-3/2})^2$.

We now turn to proving that the confidence can be raised to at least
$1-\delta$ for $\delta \in (0, 1]$ by invoking {\ApproxMCCore}
$\mathcal{O}(\log_2(1/\delta))$ times, and by using the median of the
non-$\bot$ counts thus returned.  For convenience of exposition, we
use $\eta(t, m, p)$ in the following discussion to denote the
probability of at least $m$ heads in $t$ independent tosses of a
biased coin with $\prob\left[\mathit{heads}\right] = p$.  Clearly, 
$\eta(t, m, p) = \sum_{k=m}^{t} \binom{t}{k} p^{k} (1-p)^{t-k}$.

\begin{theorem} \label{theorem:approx}
Given a propositional formula $F$ and parameters $\varepsilon ~(0 <
\varepsilon \le 1)$ and $\delta ~(0 < \delta \le
1)$, suppose {\ApproxMC}$(F, \varepsilon, \delta)$ returns $c$.  Then
$\prob\left[{\left(1+\varepsilon\right)}^{-1}\cdot |R_F| \le c \right.$ $\left.\le
  (1+\varepsilon)\cdot |R_F|\right]$ $\ge 1-\delta$.
\end{theorem}
\begin{proof}
Throughout this proof, we assume that {\ApproxMCCore} is invoked $t$
times from {\ApproxMC}, where $t = \left\lceil 35\log_2 (3/\delta)
\right\rceil$ (see pseudocode for {\ComputeIterCount} in
Section~\ref{sec:algo}).  Referring to the pseudocode of {\ApproxMC},
the final count returned by {\ApproxMC} is the median of non-$\bot$
counts obtained from the $t$ invocations of {\ApproxMCCore}.  Let
$Err$ denote the event that the median is not in
$\left[(1+\varepsilon)^{-1}\cdot |R_F|, (1+\varepsilon)\cdot |R_F|\right]$.  Let
``$\#\mathit{non }\bot = q$'' denote the event that $q$ (out of $t$)
values returned by {\ApproxMCCore} are non-$\bot$.  Then,
$\prob\left[Err\right]$ $=$ $\sum_{q=0}^t \prob\left[Err \mid
  \#\mathit{non }\bot = q\right]$ $\cdot$
$\prob\left[\#\mathit{non }\bot = q\right]$.

In order to obtain $\prob\left[Err \mid \#\mathit{non }\bot =
  q\right]$, we define a $0$-$1$ random variable $Z_i$, for $1 \le i
\le t$, as follows.  If the $i^{th}$ invocation of {\ApproxMCCore}
returns $c$, and if $c$ is either $\bot$ or a non-$\bot$ value that
does not lie in the interval $[(1+\varepsilon)^{-1}\cdot |R_F|,
  (1+\varepsilon)\cdot |R_F|]$, we set $Z_i$ to 1; otherwise, we set
it to $0$.  From Theorem~\ref{thm:almost-approx}, $\prob\left[Z_i =
  1\right] = p < 0.4$.  If $Z$ denotes $\sum_{i=1}^t Z_i$, a necessary
(but not sufficient) condition for event $Err$ to occur, given that
$q$ non-$\bot$s were returned by {\ApproxMCCore}, is $Z \ge
(t-q+\lceil q/2\rceil)$.  To see why this is so, note that $t-q$
invocations of {\ApproxMCCore} must return $\bot$.  In addition, at
least $\lceil q/2 \rceil$ of the remaining $q$ invocations must return
values outside the desired interval. To simplify the exposition, let
$q$ be an even integer.  A more careful analysis removes this
restriction and results in an additional constant scaling factor for
$\prob\left[Err\right]$.  With our simplifying assumption,
$\prob\left[Err \mid \#\mathit{non }\bot = q\right] \le \prob[Z \ge (t
  - q + q/2)]$ $=\eta(t, t-q/2, p)$.  Since $\eta(t, m, p)$ is a
decreasing function of $m$ and since $q/2 \le t-q/2 \le t$, we have
$\prob\left[Err \mid \#\mathit{non }\bot = q\right] \le \eta(t, t/2,
p)$.  If $p < 1/2$, it is easy to verify that $\eta(t, t/2, p)$ is an
increasing function of $p$.  In our case, $p < 0.4$; hence,
$\prob\left[Err \mid \#\mathit{non }\bot = q\right] \le \eta(t, t/2,
0.4)$.

It follows from above that $\prob\left[ Err \right]$ $=$
$\sum_{q=0}^t$ $\prob\left[Err \mid \#\mathit{non }\bot = q\right]$
$\cdot\prob\left[\#\mathit{non }\bot = q\right]$ $\le$ $\eta(t, t/2,
0.4)\cdot$ $\sum_{q=0}^t \prob\left[\#\mathit{non }\bot = q\right]$
$=$ $\eta(t, t/2, 0.4)$.  Since $\binom{t}{t/2} \ge \binom{t}{k}$ for
all $t/2 \le k \le t$, and since $\binom{t}{t/2} \le 2^t$, we have
$\eta(t, t/2, 0.4)$ $=$ $\sum_{k=t/2}^{t} \binom{t}{k} (0.4)^{k}
(0.6)^{t-k}$ $\le$ $\binom{t}{t/2} \sum_{k=t/2}^t (0.4)^k (0.6)^{t-k}$
\added{$\le 2^t \sum_{k=t/2}^t (0.6)^t (0.4/0.6)^{k}$}$\le 2^t \cdot 3 \cdot (0.6 \times 0.4)^{t/2}$ $\le 3\cdot(0.98)^t$.
Since $t = \left\lceil 35\log_2 (3/\delta) \right\rceil$, it follows
that $\prob\left[Err\right] \le \delta$.
\end{proof}

\begin{theorem} \label{th:complexity}
Given an oracle for {\SAT}, {\ApproxMC}$(F, \varepsilon, \delta)$
runs in time polynomial in $\log_2(1/\delta), |F|$ and
$1/\varepsilon$ relative to the oracle.
\end{theorem}
\begin{proof}
Referring to the pseudocode for {\ApproxMC}, lines $1$--$3$ take time
no more than a polynomial in $\log_2(1/\delta)$ and $1/\varepsilon$.
The repeat-until loop in lines $4$--$9$ is repeated $t = \left\lceil
35\log_2(3/\delta)\right\rceil$ times. The time taken for each
iteration is dominated by the time taken by {\ApproxMCCore}.  Finally,
computing the median in line $10$ takes time linear in $t$.  The proof
is therefore completed by showing that {\ApproxMCCore} takes time
polynomial in $|F|$ and $1/\varepsilon$ relative to the {\SAT} oracle.

Referring to the pseudocode for {\ApproxMCCore}, we find that
{\BoundedSAT} is called $\mathcal{O}(|F|)$ times.  Each such call can
be implemented by at most $\mathit{pivot}+1$ calls to a {\SAT} oracle,
and takes time polynomial in $|F|$ and $\mathit{pivot}+1$ relative to
the oracle.  Since $\mathit{pivot}+1$ is in
$\mathcal{O}(1/\varepsilon^2)$, the number of calls to the {\SAT}
oracle, and the total time taken by all calls to {\BoundedSAT} in each
invocation of {\ApproxMCCore} is a polynomial in $|F|$ and
$1/\varepsilon$ relative to the oracle.  The random choices in lines 8
and 9 of {\ApproxMCCore} can be implemented in time polynomial in $n$
(hence, in $|F|$) if we have access to a source of random bits.
Constructing $F \wedge h(z_1, \ldots z_n) = \alpha$ in line 10 can
also be done in time polynomial in $|F|$.
\end{proof}

\section{Experimental Methodology}
\label{sec:experiment}
To evaluate the performance and quality of results of {\ApproxMC}, we
built a prototype implementation and conducted an extensive set of
experiments. The suite of benchmarks represent problems from practical domains as well as problems of theoretical interest. In particular, we considered a wide range of model counting benchmarks
from different domains including grid networks, plan recognition, DQMR networks, Langford sequences, circuit
synthesis, random $k$-CNF and logistics problems ~\cite{SangBearKautz2005,KrocSabSel2008}. The suite consisted of benchmarks ranging from 32 variables to 229100 variables in CNF representation. The complete set of benchmarks (numbering
above $200$) is available at \url{http://www.cs.rice.edu/CS/Verification/Projects/ApproxMC/}.  

All our experiments were conducted on a high-performance computing
cluster.  Each individual experiment was run on a single node of the
cluster; the cluster allowed multiple experiments to run in parallel.
Every node in the cluster had two quad-core Intel Xeon processors with
$4$GB of main memory. We used $2500$ seconds as the timeout for each
invocation of {\BoundedSAT} in {\ApproxMCCore}, and $20$ hours as the
timeout for {\ApproxMC}. If an invocation of {\BoundedSAT} in line
$10$ of the pseudo-code of {\ApproxMCCore} timed out, we repeated the
iteration (lines 6-11 of the pseudocode of {\ApproxMCCore}) without
incrementing $i$. The parameters $\varepsilon$ (tolerance) and
$\delta$ (confidence being $1-\delta$) were set to $0.75$ and $0.1$
respectively.  With these parameters, {\ApproxMC} successfully
computed counts for benchmarks with upto $33,000$ variables.

We implemented leap-frogging, as described in~\cite{SKV13}, to
estimate initial values of $i$ from which to start iterating the
repeat-until loop of lines $6$--$11$ of the pseudocode of
{\ApproxMCCore}.  %
To further optimize the running time, we obtained tighter estimates of
the iteration count $t$ used in algorithm {\ApproxMC}, compared to
those given by algorithm {\ComputeIterCount}.  A closer examination of
the proof of Theorem~\ref{theorem:approx} shows that it suffices to
have $\eta(t, t/2, 0.4) \le \delta$.  We therefore pre-computed a
table that gave the smallest $t$ as a function of $\delta$ such that
$\eta(t, t/2, 0.4) \le \delta$.  This sufficed for all our experiments
and gave smaller values of $t$ (we used $t$=41 for $\delta$=0.1) compared to those given by
{\ComputeIterCount}.

For purposes of comparison, we also implemented and conducted
experiments with the exact counter
{\Cachet}~\cite{Sang04combiningcomponent} by setting a timeout of $20$
hours on the same computing platform.  We compared the running time of
{\ApproxMC} with that of {\Cachet} for several benchmarks, ranging
from benchmarks on which {\Cachet} ran very efficiently to those on
which {\Cachet} timed out.  We also measured the quality of
approximation produced by {\ApproxMC} as follows.  For each benchmark
on which {\Cachet} did not time out, we obtained the approximate count
from {\ApproxMC} with parameters $\varepsilon = 0.75$ and $\delta =
0.1$, and checked if the approximate count was indeed within a factor
of $1.75$ from the exact count.  Since the theoretical guarantees
provided by our analysis are conservative, we also measured the
relative error of the counts reported by {\ApproxCount} using the
$L_1$ norm, for all benchmarks on which {\Cachet} did not time out.
For an input formula $F_i$, let $A_{F_i}$ (resp., $C_{F_i}$) be the
count returned by {\ApproxCount} (resp., {\Cachet}).  We computed the
$L_1$ norm of the relative error as $\frac{\sum_i |A_{F_i} -
  C_{F_i}|}{\sum_i C_{F_{i}}}$.
     
Since {\Cachet} timed out on most large benchmarks, we compared
{\ApproxMC} with state-of-the-art bounding counters as well.  As
discussed in Section~\ref{sec:intro}, bounding counters do not provide
any tolerance guarantees.  Hence their guarantees are significantly
weaker than those provided by {\approxMC}, and a direct comparison of
performance is not meaningful.  Therefore, we compared the sizes of
the intervals (i.e., difference between upper and lower bounds)
obtained from existing state-of-the-art bounding counters with those
obtained from {\approxMC}.  To obtain intervals from {\ApproxMC}, note
that Theorem~\ref{theorem:approx} guarantees that if {\ApproxMC}$(F,
\varepsilon, \delta)$ returns $c$, then $\prob[\frac{c}{1+\varepsilon}
  \le |R_F| $ $\le (1+\varepsilon)\cdot c]$ $\ge 1-\delta$.
Therefore, {\ApproxMC} can be viewed as computing the interval
$[\frac{c}{1+\varepsilon}, (1+\varepsilon)\cdot c]$ for the model count,
with confidence $\delta$.  We considered state-of-the-art lower
bounding counters, viz.  {\MBound}~\cite{gomes2006model},
{\HybridMBound}~\cite{gomes2006model},
{\SampleCount}~\cite{gomes2007sampling} and
{\BPCount}~\cite{KrocSabSel2008}, to compute a lower bound of the
model count, and used {\MiniCount}~\cite{KrocSabSel2008} to obtain an
upper bound.  We observed that {\SampleCount} consistently produced
better (i.e. larger) lower bounds than {\BPCount} for our benchmarks.
Furthermore, the authors of~\cite{gomes2006model} advocate using
{\HybridMBound} instead of {\MBound}.  Therefore, the lower bound for
each benchmark was obtained by taking the maximum of the bounds
reported by {\HybridMBound} and {\SampleCount}.

We set the confidence value for {\MiniCount} to $0.99$ and {\SampleCount} and {\HybridMBound} to $0.91$. For a detailed 
justification of these choices, we refer the reader to the full version of our paper.
Our implementation of {\HybridMBound} used the ``conservative"
approach described in~\cite{gomes2006model}, since this provides the
best lower bounds with the required confidence among all the
approaches discussed in~\cite{gomes2006model}.  Finally, to ensure
fair comparison, we allowed all bounding counters to run for $20$
hours on the same computing platform on which {\ApproxMC} was run.
\vspace*{-.15in}

\section{Results}
\label{sec:results}
The results on only a subset of our benchmarks are presented here for
lack of space.  
\begin{figure}[h!b]
\centering
\includegraphics[scale=0.7]{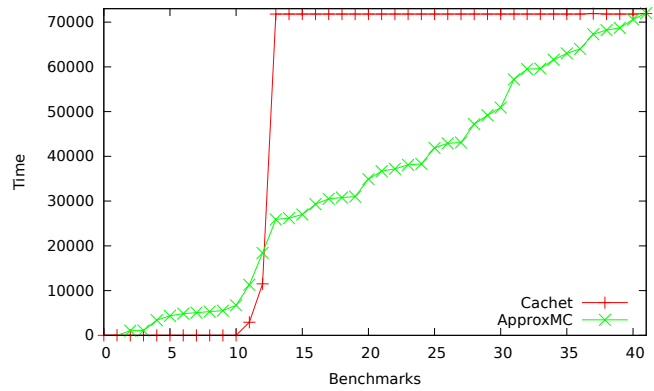}
\caption{Performance comparison between {\ApproxMC} and {\Cachet}. The
  benchmarks are arranged in increasing order of running time of
  {\ApproxMC}.}
\label{fig:perfomance_comparisons}
\end{figure}
Figure~\ref{fig:perfomance_comparisons} shows how the running times of
{\ApproxMC} and {\Cachet} compared on this subset of our benchmarks.  The
y-axis in the figure represents time in seconds, while the x-axis
represents benchmarks arranged in ascending order of running time of
{\ApproxMC}.  The comparison shows that although {\Cachet} performed
better than {\ApproxMC} initially, it timed out as the ``difficulty"
of problems increased.  {\ApproxMC}, however, continued to return
bounds with the specified tolerance and confidence, for many more
difficult and larger problems.  Eventually, however, even {\ApproxMC}
timed out for very large problem instances.  Our experiments clearly
demonstrate that there is a large class of practical problems that lie
beyond the reach of exact counters, but for which we can still obtain
counts with $(\varepsilon, \delta)$-style guarantees in reasonable
time.  This suggests that given a model counting problem, it is
advisable to run {\Cachet} initially with a small timeout.  If
{\Cachet} times out, {\ApproxMC} should be run with a larger timeout.
Finally, if {\ApproxMC} also times out, counters with much weaker
guarantees but shorter running times, such as bounding counters,
should be used.

\begin{figure}[h!b]\centering
\includegraphics[scale=0.7]{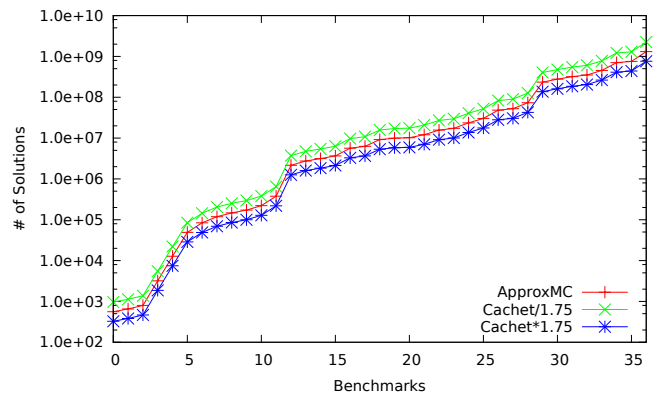}
\caption{Quality of counts computed by {\approxMC}. The benchmarks are
  arranged in increasing order of model counts.}
\label{fig:quality_comparison}
\end{figure}
\begin{figure}[h!b]
\centering
\includegraphics[scale=0.7]{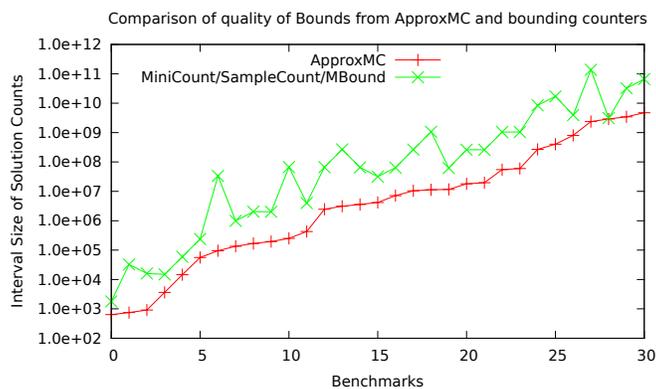}
\caption{Comparison of interval sizes from {\approxMC} and those from
  bounding counters. The benchmarks are arranged in increasing order
  of model counts.}
\label{fig:bounding_comparisons}
\end{figure}
Figure \ref{fig:quality_comparison} compares the model count computed
by {\ApproxMC} with the bounds obtained by scaling the exact count
obtained from {\Cachet} by the tolerance factor ($1.75$) on
a subset of our benchmarks. The y-axis in this figure represents the
model count on a log-scale, while the x-axis represents the
benchmarks arranged in ascending order of the model count.
The figure shows that in all cases, the count reported by {\ApproxMC}
lies within the specified tolerance of the exact count.  Although we
have presented results for only a subset of our benchmarks ($37$
in total) in Figure~\ref{fig:quality_comparison} for reasons of
clarity, the counts reported by {\ApproxMC} were found to be within
the specified tolerance of the exact counts for \emph{all} $95$
benchmarks for which {\Cachet} reported exact counts.  %
We also found that the $L_1$ norm of the relative
error, considering all $95$ benchmarks for which {\Cachet} returned
exact counts, was $0.033$.  Thus, {\ApproxMC} has approximately $4\%$
error in practice -- much smaller than the theoretical guarantee of
$75\%$ with $\varepsilon = 0.75$.
  
Figure~\ref{fig:bounding_comparisons} compares the sizes of intervals
computed using {\ApproxMC} and using state-of-the-art bounding
counters (as described in Section~\ref{sec:experiment}) on a subset of
our benchmarks.  The comparison clearly shows that the sizes of
intervals computed using {\ApproxMC} are consistently smaller than the
sizes of the corresponding intervals obtained from existing bounding
counters.  Since smaller intervals with comparable confidence
represent better approximations, we conclude that {\approxMC} computes
better approximations than a combination of existing bounding
counters. In all cases, {\approxMC} improved the upper bounds from {\MiniCount} significantly; it also improved lower bounds from {\SampleCount} and {\MBound} to a lesser extent. For details, please refer to the full version. 

\section{Conclusion and Future Work}\label{sec:discussion}
We presented {\approxMC}, the first $(\varepsilon, \delta)$
approximate counter for CNF formulae that scales in practice to tens
of thousands of variables.  We showed that {\approxMC} reports bounds
with small tolerance in theory, and with much smaller error in
practice, with high confidence. Extending the ideas in this paper to
probabilistic inference and to count models of SMT constraints is an interesting direction of future research. 

\clearpage
\bibliography{Report}

\begin{thebibliography}{10}

\bibitem{CryptoMiniSAT}
{CryptoMiniSAT}.
\newblock \url{http://www.msoos.org/cryptominisat2/}.

\bibitem{HotBits}
{HotBits}.
\newblock \url{http://www.fourmilab.ch/hotbits}.

\bibitem{Angluin1980}
D.~Angluin.
\newblock On counting problems and the polynomial-time hierarchy.
\newblock {\em Theoretical Computer Science}, 12(2):161 -- 173, 1980.

\bibitem{Bacchus04algorithmsand}
F.~Bacchus, S.~Dalmao, and T.~Pitassi.
\newblock Algorithms and complexity results for \#{SAT} and bayesian inference.
\newblock In {\em Proc. of {FOCS}}, pages 340--351, 2004.

\bibitem{Bellare98uniformgeneration}
M.~Bellare, O.~Goldreich, and E.~Petrank.
\newblock Uniform generation of {NP}-witnesses using an {NP}-oracle.
\newblock {\em Information and Computation}, 163(2):510--526, 1998.

\bibitem{Birnbaum1999}
E.~Birnbaum and E.~L. Lozinskii.
\newblock The good old {Davis-Putnam} procedure helps counting models.
\newblock {\em Journal of Artificial Intelligence Research}, 10(1):457--477,
  June 1999.

\bibitem{SKV13}
S.~Chakraborty, K.S. Meel, and M.Y. Vardi.
\newblock A scalable and nearly uniform generator of {SAT} witnesses.
\newblock In {\em Proc. of CAV}, 2013.

\bibitem{darwiche2004new}
A.~Darwiche.
\newblock New advances in compiling {CNF} to decomposable negation normal form.
\newblock In {\em Proc. of ECAI}, pages 328--332. Citeseer, 2004.

\bibitem{domshlak2007}
C.~Domshlak and J.~Hoffmann.
\newblock Probabilistic planning via heuristic forward search and weighted
  model counting.
\newblock {\em Journal of Artificial Intelligence Research}, 30(1):565--620,
  2007.

\bibitem{ErGoSel2012}
S.~Ermon, C.P. Gomes, and B.~Selman.
\newblock Uniform solution sampling using a constraint solver as an oracle.
\newblock In {\em Proc. of UAI}, 2012.

\bibitem{GogDech2011}
V.~Gogate and R.~Dechter.
\newblock Samplesearch: Importance sampling in presence of determinism.
\newblock {\em Artificial Intelligence}, 175(2):694--729, 2011.

\bibitem{GSS08}
C.~P. Gomes, A.~Sabharwal, and B.~Selman.
\newblock Model counting.
\newblock In A.~Biere, M.~Heule, H.~V. Maaren, and T.~Walsh, editors, {\em
  Handbook of Satisfiability}, volume 185 of {\em Frontiers in Artificial
  Intelligence and Applications}, pages 633--654. IOS Press, 2009.

\bibitem{gomes2006model}
Carla~P. Gomes, Ashish Sabharwal, and Bart Selman.
\newblock Model counting: A new strategy for obtaining good bounds.
\newblock In {\em Proc. of AAAI}, pages 54--61, 2006.

\bibitem{gomes2007sampling}
C.P. Gomes, J.~Hoffmann, A.~Sabharwal, and B.~Selman.
\newblock From sampling to model counting.
\newblock In {\em Proc. of IJCAI}, pages 2293--2299, 2007.

\bibitem{Gomes-Sampling}
C.P. Gomes, A.~Sabharwal, and B.~Selman.
\newblock Near-uniform sampling of combinatorial spaces using {XOR}
  constraints.
\newblock In {\em Proc. of NIPS}, pages 670--676, 2007.

\bibitem{Jerr}
M.R. Jerrum, L.G. Valiant, and V.V. Vazirani.
\newblock Random generation of combinatorial structures from a uniform
  distribution.
\newblock {\em Theoretical Computer Science}, 43(2-3):169--188, 1986.

\bibitem{Bayardo97usingcsp}
R.~J.~Bayardo Jr. and R.~Schrag.
\newblock Using {CSP} look-back techniques to solve real-world {SAT} instances.
\newblock In {\em Proc. of AAAI}, pages 203--208, 1997.

\bibitem{KarpLuby1989}
R.M. Karp, M.~Luby, and N.~Madras.
\newblock {M}onte-{C}arlo approximation algorithms for enumeration problems.
\newblock {\em Journal of Algorithms}, 10(3):429--448, 1989.

\bibitem{KitKue2007}
N.~Kitchen and A.~Kuehlmann.
\newblock Stimulus generation for constrained random simulation.
\newblock In {\em Proc. of ICCAD}, pages 258--265, 2007.

\bibitem{KrocSabSel2008}
L.~Kroc, A.~Sabharwal, and B.~Selman.
\newblock Leveraging belief propagation, backtrack search, and statistics for
  model counting.
\newblock In {\em Proc. of CPAIOR}, pages 127--141, 2008.

\bibitem{lobbing1996}
M.~L\"{o}bbing and I.~Wegener.
\newblock The number of knight's tours equals 33,439,123,484,294 -- counting
  with binary decision diagrams.
\newblock {\em The Electronic Journal of Combinatorics}, 3(1):R5, 1996.

\bibitem{LubyThesis83}
M.G. Luby.
\newblock {\em {M}onte-{C}arlo Methods for Estimating System Reliability}.
\newblock PhD thesis, EECS Department, University of California, Berkeley, Jun
  1983.

\bibitem{minato93}
S.~Minato.
\newblock Zero-suppressed bdds for set manipulation in combinatorial problems.
\newblock In {\em Proc. of Design Automation Conference}, pages 272--277, 1993.

\bibitem{Roth1996}
D.~Roth.
\newblock On the hardness of approximate reasoning.
\newblock {\em Artificial Intelligence}, 82(1):273--302, 1996.

\bibitem{Rubin2012}
R.~Rubinstein.
\newblock Stochastic enumeration method for counting np-hard problems.
\newblock {\em Methodology and Computing in Applied Probability}, pages 1--43,
  2012.

\bibitem{Sang04combiningcomponent}
T.~Sang, F.~Bacchus, P.~Beame, H.~Kautz, and T.~Pitassi.
\newblock Combining component caching and clause learning for effective model
  counting.
\newblock In {\em Proc. of SAT}, 2004.

\bibitem{SangBearKautz2005}
T.~Sang, P.~Bearne, and H.~Kautz.
\newblock Performing bayesian inference by weighted model counting.
\newblock In {\em Prof. of AAAI}, pages 475--481, 2005.

\bibitem{Schmidt}
J.~P. Schmidt, A.~Siegel, and A.~Srinivasan.
\newblock {Chernoff-Hoeffding} bounds for applications with limited
  independence.
\newblock {\em SIAM Journal on Discrete Mathematics}, 8:223--250, May 1995.

\bibitem{simon77}
J.~Simon.
\newblock On the difference between one and many.
\newblock In {\em Proc. of ICALP}, pages 480--491, 1977.

\bibitem{Sipser83}
M.~Sipser.
\newblock A complexity theoretic approach to randomness.
\newblock In {\em Proc. of STOC}, pages 330--335, 1983.

\bibitem{Stockmeyer83}
L.~Stockmeyer.
\newblock The complexity of approximate counting.
\newblock In {\em Proc. of STOC}, pages 118--126, 1983.

\bibitem{Thurley2006}
M.~Thurley.
\newblock {sharpSAT}: counting models with advanced component caching and
  implicit bcp.
\newblock In {\em Proc. of {SAT}}, pages 424--429, 2006.

\bibitem{toda1989computational}
S.~Toda.
\newblock On the computational power of {PP} and {(+)P}.
\newblock In {\em Proc. of FOCS}, pages 514--519. IEEE, 1989.

\bibitem{trevisan2002lecture}
L.~Trevisan.
\newblock Lecture notes on computational complexity.
\newblock {\em Notes written in Fall}, 2002.
\newblock
  \url{http://citeseerx.ist.psu.edu/viewdoc/download?doi=10.1.1.71.9877&rep=rep1&type=pdf}.

\bibitem{valiant1979complexity}
L.G. Valiant.
\newblock The complexity of enumeration and reliability problems.
\newblock {\em SIAM Journal on Computing}, 8(3):410--421, 1979.

\bibitem{wei2005new}
W.~Wei and B.~Selman.
\newblock A new approach to model counting.
\newblock In {\em Proc. of SAT}, pages 2293--2299. Springer, 2005.

\bibitem{Yuan2004}
J.~Yuan, A.~Aziz, C.~Pixley, and K.~Albin.
\newblock Simplifying boolean constraint solving for random simulation-vector
  generation.
\newblock {\em IEEE Trans. on CAD of Integrated Circuits and Systems},
  23(3):412--420, 2004.

\end{thebibliography}
\bibliographystyle{plain}  
\clearpage
\end{document}